\documentclass[11pt]{article}
\usepackage{blindtext}
\usepackage[usenames,dvipsnames,svgnames,table]{xcolor}
\usepackage{enumitem}
\usepackage{caption}
\usepackage{placeins}
\usepackage{graphicx}

\usepackage{fancybox}
\usepackage{hyperref}

\usepackage{amsmath,amsfonts,amsthm,amssymb,xcolor}
\usepackage{bm} 

\usepackage{nicefrac}

\usepackage{float}

\newtheorem{theorem}{Theorem}[section]

\newtheorem{corollary}[theorem]{Corollary}
\newtheorem{lemma}[theorem]{Lemma}

\newtheorem{proposition}[theorem]{Proposition}
\newtheorem{claim}[theorem]{Claim}

\newtheorem{fact}[theorem]{Fact}

\theoremstyle{definition}
\newtheorem{definition}[theorem]{Definition}

\newtheorem*{theorem*}{Theorem}
\newtheorem*{corollary*}{Corollary}
\newtheorem*{conjecture*}{Conjecture}
\newtheorem*{lemma*}{Lemma}
\newtheorem*{thm*}{Theorem}
\newtheorem*{prop*}{Proposition}
\newtheorem*{obs*}{Observation}
\newtheorem*{definition*}{Definition}

\newtheorem*{remark*}{Remark}
\newtheorem*{rec*}{Recommendation}

\newenvironment{fminipage}%
  {\begin{Sbox}\begin{minipage}}%
  {\end{minipage}\end{Sbox}\fbox{\TheSbox}}

\def\prob#1#2{\mbox{Pr}_{#1}\left[ #2 \right]}

\def\expec#1#2{{\mathbb{E}}_{#1}\left[ #2 \right]}

\def\setof#1{\left\{#1  \right\}}

\def\union{\cup}

\def\abs#1{\left|#1  \right|}

\def\norm#1{\left\| #1 \right\|}

\def\calD{\mathcal{D}}

\def\calL{\mathcal{L}}

\def\calT{\mathcal{T}}
\def\calM{\mathcal{M}}
\def\calK{\mathcal{K}}

\newcommand\LLambda{\boldsymbol{\mathit{\Lambda}}}
\newcommand\PPi{\boldsymbol{\Pi}}

\newcommand\xxi{\boldsymbol{\xi}}

\newcommand\bb{\boldsymbol{\mathit{b}}}

\newcommand\hh{\boldsymbol{\mathit{h}}}

\newcommand\pp{\boldsymbol{\mathit{p}}}
\newcommand\qq{\boldsymbol{\mathit{q}}}

\newcommand\yy{\boldsymbol{\mathit{y}}}

\newcommand\xx{\boldsymbol{\mathit{x}}}

\newcommand\veczero{\boldsymbol{0}}
\newcommand\vecone{\boldsymbol{1}}

\newcommand\matzero{\boldsymbol{0}}

\renewcommand\AA{\boldsymbol{\mathit{A}}}
\newcommand\BB{\boldsymbol{\mathit{B}}}
\newcommand\CC{\boldsymbol{\mathit{C}}}

\newcommand\HH{\boldsymbol{{H}}}
\newcommand\II{\boldsymbol{\mathit{I}}}

\newcommand\MM{\boldsymbol{\mathit{M}}}
\newcommand\LL{\boldsymbol{\mathit{L}}}

\newcommand\VV{\boldsymbol{\mathit{V}}}
\newcommand\XX{\boldsymbol{\mathit{X}}}
\newcommand\YY{\boldsymbol{\mathit{Y}}}

\newcommand\ZZ{\boldsymbol{\mathit{Z}}}

\newcommand\llambda{\boldsymbol{\mathit{\lambda}}}

\newcommand\R{\mathbb{R}}

\newcommand\C{\mathbb{C}}

\newcommand{\E}[1]{\mathop{{}\mathbb{E}}\left[#1\right]}
\newcommand{\Var}[1]{\mathop{{}\operatorname{Var}}\left[#1\right]}

\newcommand\tr{\mathrm{Tr}}

\newcommand{\trp}{\top}

\DeclareMathOperator*{\diag}{diag}

\usepackage{fullpage}
\usepackage[letterpaper,
            left=1in,right=1in,top=1in,bottom=1in]{geometry}
\usepackage{amsmath}
\usepackage{amssymb}
\usepackage{amsthm}
\usepackage{graphicx}
\usepackage{float}
\usepackage{color}
\usepackage{pifont}
\usepackage[ruled, linesnumbered]{algorithm2e}
\usepackage[
  backend=biber,
  backref=true,
  backrefstyle=none,
  date=year,
  doi=false,
  giveninits=true,
  hyperref=true,
  maxbibnames=10,
  sortcites=false,
  style=alphabetic,
  url=false, 
]{biblatex}
\def\calI{\mathcal{I}}

\def\iindependent{$\ell_{\infty}$-independent}
\def\iindependence{$\ell_{\infty}$-independence}
\def\amindependent{average multiplicatively independent}
\def\amindependence{average multiplicative independence}

\def\Dinf{D_{\text{inf}}}
\def\Dam{D_{\text{am}}}
\def\tsiindependent{two-sided \iindependent}
\def\tsiindependence{two-sided \iindependence}

\def\muhom{\mu^{\text{hom}}}

\DeclareMathOperator{\dist}{dist}

\newif\ifsodaversion
\sodaversionfalse

\begin{document}

\title{Scalar and Matrix Chernoff Bounds from $\ell_{\infty}$-Independence}
\author{
  Tali Kaufman \\ Bar-Ilan University \\ \href{mailto:kaufmant@mit.edu}{kaufmant@mit.edu} 
  \and
  Rasmus Kyng\\ ETH Zurich \\ \href{mailto:kyng@inf.ethz.ch}{kyng@inf.ethz.ch} 
  \and
  Federico Sold\'a \\ ETH Zurich\\ \href{mailto:federico.solda@inf.ethz.ch}{federico.solda@inf.ethz.ch} 
}
\date{\today}

\clearpage\maketitle
\thispagestyle{empty}
\vspace{-1em}
\begin{abstract}
We present new scalar and matrix Chernoff-style concentration bounds for a broad class of probability distributions over the binary hypercube $\{0,1\}^n$.
Motivated by recent tools developed for the study of mixing times of Markov chains on discrete distributions, we say that a distribution is $\ell_\infty$-independent when the infinity norm of its influence matrix $\mathcal{I}$ is bounded by a constant. We show that any distribution which is $\ell_\infty$-independent satisfies a matrix Chernoff bound that matches the matrix Chernoff bound for independent random variables due to Tropp. Our matrix Chernoff bound is a broad generalization and strengthening of the matrix Chernoff bound of Kyng and Song (FOCS'18). Using our bound, we can conclude as a corollary that a union of $O(\log|V|)$ random spanning trees gives a spectral graph sparsifier of a graph with $|V|$ vertices with high probability, matching results for independent edge sampling, and matching lower bounds from Kyng and Song.

\end{abstract}


\newpage
\pagenumbering{arabic} 

\section{Introduction}
\paragraph{Concentration of dependent scalar random variables.}
One of the most important and fundamental properties of random variables is that in many settings, they ``concentrate'' around some expected outcome, i.e. they are close to some typical outcome.
The central limit theorem provides a limiting version of this result, while standard Chernoff bounds for sums of independent Bernoulli random variables give one of the most important examples of non-asymptotic concentration statements.

Many concentration statements focus on random variables constructed in some way from a collection of other independent random variables, e.g. by taking a sum of these independent r.v.s or by considering another function whose inputs are independent random variables.
However, in many applications, we want to understand the concentration properties of random variables constructed from a collection of other variables that are \emph{not independent}.
A classical example in theoretical computer science is a random spanning tree of an undirected (weighted) graph.
To describe a random spanning tree in a graph $G = (V,E)$ with $\abs{E}$ edges and $\abs{V}$ vertices, we can construct a binary vector in $\setof{0,1}^{\abs{E}}$ with $\abs{V}-1$ entries that equal 1, indicating the edges that are present in the tree.
These $\abs{E}$ random Bernoulli variables are not independent, but they still exhibit many forms of concentration.

An important result by Dubhashi and Ranjan \cite{DR98} showed that standard Chernoff bounds apply to random binary vectors whose variables exhibit various forms of negative correlation, in particular if so-called \emph{negative regression} or \emph{negative association} holds for the distribution.
As a consequence, Chernoff bounds apply to the indicator vector for a random spanning tree as described above.
Beyond negative correlation, Paulin \cite{P14} proved standard Chernoff bounds for weakly dependent random variables satisfying a condition known as Dobrushin uniqueness condition.

A random vector $\xxi \in \setof{0,1}^n$ is called $k$-homogeneous if every outcome of $\xxi$ contains exactly $k$ ones. $k$-homogeneous distributions are stationary distributions of some natural random walks on pure simplicial complexes namely the ``up and down'' and the ``down and up'' walks. As part of a new theory of high-dimensional expanders, Kaufman and Mass introduced the study of these high order random walks \cite{KM17}. Subsequent works of 
Dinur, Kaufman and Oppenheim \cite{DK17, KO18, O18} provided tight spectral analysis of these walks that implies fast mixing time bounds for these random walks. This, in turn, makes it possible to efficiently sample (approximately) from their stable distribution.
Building on this theory, \cite{ALGV19a} proved a breakthrough result on rapid mixing for distributions associated with log-concave polynomials that include the uniform distribution on the bases of matroids \cite{AHK18,AGV18a} and strongly Rayleigh distributions \cite{BBL09a}.
The result was strengthened and the proof simplified by Cryan et al. \cite{CGM19}  and yet further improvements in the mixing time were obtained by Anari et al.\ \cite{ALGVV21}.
Further developing this theory, Anari et al.\ \cite{ALG21} introduced the spectral independence property and showed that it implies rapid mixing. The spectral independence approach has been used by Anari et al.\ \cite{ALG21} and several subsequent works \cite{AASV21a, ALG21, BCCPSV21, CGSV21, CLV20, CLV21b, FGYZ21, JPV21, L21a} to prove rapid mixing of an important case of ``down and up'' walk called Glauber dynamics in particular in the framework of spin systems. Spin systems capture many combinatorial models of interest, including the hard-core model on weighted independent sets, the Ising model, and colorings.
In this context, for some regimes, optimal mixing times has been obtained by Chen et al. \cite{CLV21} using approximate tensorization of entropy and generalizing the work of Cryan et al. \cite{CGM19}.  Recently, Anari at al. \cite{AJKPV21} introduced the notion of entropic independence as an analog of spectral independence and established the tight mixing time for Glauber dynamics for a broad range of Ising models.
It is possible to extend many of these analyses to non-homogeneous distributions by using a \emph{homogenization argument} which pads the random vector with enough entries to enforce homogeneity. This, however, frequently leads to much weaker bounds than analyses that directly work on non-homogeneous distributions.

A priori, it is not clear that proving fast mixing of such random walks proves concentration of some sort for samples from $\xxi$.
But, a central step in proving mixing is to bound various quantities known as the Poincar\'e or (modified) log-Sobolev constants of the walk.
It turns out that somewhat standard techniques such as the ``Herbst argument'' can translate a bound on these mixing-time quantities into a concentration statement (see \cite{BLM13} for a concentration argument from the modified log-Sobolev inequality and see \cite{L99} for a concentration argument from the Poicar\'e inequality).
When using the stronger modified log-Sobolev inequality, these arguments show that for a $k$-homogeneous random vector $\xxi \in \setof{0,1}^n$, given a function $f : \setof{0,1}^n \to \R$ that is $1$-Lipschitz w.r.t. Hamming distance,
$\abs{f(\xxi) - \E{f(\xxi)}} \leq O(\sqrt{k \log(1/\delta)})$ with probability at least $1-\delta$.
We call this type of concentration a McDiarmid bound as an early variant of this result was shown by \cite{M89}.
For distributions that have the \emph{Stochastic Covering Property} (SCP), Peres and Pemantle \cite{PP14} showed precisely this bound -- and the indicator vector of a random spanning tree falls under this category. 
By bounding the modified log-Sobolev constant, Hermon and Salez \cite{HS19} obtained again the result of Peres and Pemantle \cite{PP14}, while Cryan et al. \cite{CGM19} obtained a McDiarmid-like bound for $k$-homogeneous strongly log-concave distributions. Later work by \cite{GV18} extended this type of concentration bound to apply to the full range of negative association and negative regression type distributions considered by \cite{DR98} that include all the distributions that satisfy the SCP.
For non-homogeneous distributions, the recent Poincar\'e and MLS-based concentration approaches can be applied by first using a homogenization argument. This pads the random vector obtained by adding some extra variables so that all the outcomes contains exactly $n$ ones. In combination with McDiarmid-like bounds, this technique leads to particularly coarse bounds that depend on $n$ instead of $k$.

In contrast, Chernoff bounds apply to the more restrictive function class\footnote{Observe that the setting of Chernoff bounds also implies that $f$ is $1$-Lipschitz w.r.t. Hamming distance.} of  $f(\xxi) = \sum_i c_i \xxi(i)$ with ${c_i \in [0,1]}$ for all $i$.
The classical Chernoff bound and the matching bounds of \cite{DR98} show for various distributions that $\abs{f(\xxi) - \E{f(\xxi)}} \leq O(\sqrt{\E{f(\xxi)} \log(1/\delta)})$ with probability at least $1-\delta$, provided $\log(1/\delta) \leq \E{f(\xxi)}$.
We call this a ``Chernoff-type bound''.
This Chernoff-type bound can be much stronger than the McDiarmid bound when $\E{f(\xxi)} \ll k$.
Thus, generally, while Chernoff bounds are sometimes perceived as a weaker cousin of the more general McDiarmid bound for $1$-Lipschitz functions, the Chernoff bound can prove much stronger concentration in some cases.

\paragraph{Concentration of dependent matrix-valued random variables.}
Beyond scalar-valued random variables, an important area of study for concentration has been random matrices. In concentration theory for (symmetric) random matrices, we study a (symmetric, real-valued) matrix valued random variable, say, $\XX \in \R^{d \times d}$.
Two important cases here are analogous to the scalar sum-function and $1$-Lipschitz function settings we considered above.
In the matrix setting, we find that the gap between different types of matrix-concentration statements widens.

Rudelson \cite{R99} and Ahlswede and Winter \cite{AW02} proved early Chernoff-like bounds for matrices.
Tropp \cite{T12} proved a more refined bound, which we will focus on.
Consider a random binary vector  $\xxi \in \setof{0,1}^n$ and a function $f(\xxi) = \sum_i \CC_i \xxi(i)$ where each $\CC_i$ is a positive-semi definite $d \times d$ matrix with spectral norm $\norm{\CC_i} \leq 1$.
When the $\xxi$ entries are independent, Tropp's bound shows that
$\lambda_{\max}\left( f(\xxi) -\expec{}{f(\xxi)} \right) \leq O(\sqrt{ \lambda_{\max}(\E{f(\xxi)}) \log(d/\delta) })$ with probability at least $1-\delta$, provided $\log(d/\delta) \leq \lambda_{\max}(\E{f(\xxi)})$.
A similar bound holds for the lower tail but with $\lambda_{\min}(\cdot)$ replacing $\lambda_{\max}(\cdot)$.

For $k$-homogeneous Strongly Rayleigh distributions, \cite{KS18} proved a weaker Chernoff bound, showing that given $f$ with the same properties,
$\lambda_{\max}\left( f(\xxi) -\expec{}{f(\xxi)} \right) \leq O(\sqrt{ \log(k) \lambda_{\max}(\E{f(\xxi)}) \log(d/\delta) })$ with probability at least $1-\delta$, provided $\log(d/\delta) \leq \log(k) \lambda_{\max}(\E{f(\xxi)})$. Again, a lower tail bound also holds after replacing $\lambda_{\max}(\cdot)$ replacing $\lambda_{\min}(\cdot)$.

These bounds of \cite{T12} and \cite{KS18} that we call Chernoff-like bounds work particularly well when $\lambda_{\max}(\expec{}{f(\xxi)})$ is small, which happens when the ``mass'' of the distribution is well-spread out across all matrix ``directions''.
An example of this is matrix-valued random variables with isotropic mean or covariance matrix.
For example, in spectral graph theory, the bound \cite{T12} can be used to prove that in a graph with $\abs{E}$ edges and $\abs{V}$ vertices, we can obtain a spectral sparsifier  of the graph Laplacian using $\Theta(\abs{V} \log \abs{V})$ edges.
Similarly, \cite{KS18} shows as a corollary that a spectral sparsifier of the graph Laplacian can be constructed using $\Theta (\log^2 \abs{V})$ random spanning trees of the graph.

Auon et al. \cite{ABY20} and Kathuria \cite{K20} established matrix analogs of the Poincar\'e inequality, and 
showed that this implies concentration for symmetric-matrix-valued functions.
Garg et al. \cite{GKS21} showed the scalar Poincar\'e inequality implies the matrix Poincar\'e inequality.
Further improvements were obtained by Huang and Tropp \cite{HT21a} using the Bakry–\'Emery curvature criterion.  The bounds presented in these results are analogous to the McDiarmid bound in the scalar settings and, unfortunately, this type of bound works poorly in the important isotropic setting.
In particular when restricted to a $k$-homogeneous strongly Rayleigh distributions, for a function $f : \setof{0,1}^n \to \R^{d \times d}$, \cite{ABY20,K20,HT21a} show that  $\norm{ f(\xxi) -\expec{}{f(\xxi)} }\leq O(\sqrt{k}\log(d/\delta))$ with probability at least $1-\delta$. 
This means the bound can only prove that $O(|V|\log^2|V| )$ random spanning trees give a spectral sparsifier of a graph.
If, instead, we take some care to modify their bound to treat a union of independent SCP distributions more carefully, we can reduce this to $O(\sqrt{|V|}\log |V| )$ spanning trees to build a spectral sparsifier -- but this is still exponentially worse than \cite{KS18}.

Why are the bounds of \cite{ABY20,K20} and \cite{HT21a} weak compared to \cite{T12}, when we look at setting like spectral graph sparsification?
Somewhat heuristically, we can say that this comes from the classic matrix Chernoff (and matrix Bernstein \cite{T12}) result leading to spectral norm deviation that depends on $\lambda_{\max}(\Var{f(\xxi)})$, while the others incur spectral norm deviation that depends on $\sum_i \lambda_{i}(\Var{f(\xxi)}) = \tr(\Var{f(\xxi)})$, because they cannot tell apart the different directions of variance, and pay simultaneously for the variance in all eigenvalue directions.
We can think of this distinction as saying the classic matrix Chernoff is ``direction-aware'' whereas the \cite{ABY20,K20,HT21a} bounds are ``direction-unaware''.
In many applications of matrix concentration, such as spectral graph sparsification, a direction-aware bound is necessary to obtain good results.
Our techniques prove direction-aware bounds.

\paragraph{Our contributions: New(ish) measures of dependence and new concentration results.}
The observations above suggest that we should try to prove scalar and matrix Chernoff bounds for the many recently studied distributions.
In contrast to the well-established Herbst-argument for concentration of $1$-Lipschitz functions based on modified log-Sobolev constant or Poicar\'e constants, there is no standard recipe for proving Chernoff bounds for dependent distributions, and in fact it is not clear that Chernoff bounds hold given control over the Poincar\'e or the modified log-Sobolev constants.
Nonetheless, we show that it is possible to prove standard Chernoff bounds for many of these recently studied distributions, using a property that we call \emph{\iindependence}.
To introduce this notion, we need to recall the notion of an \emph{influence matrix} \cite{D70}, in particular, we first state the variant introduced in \cite{ALG21} that we call a \emph{two-sided} influence matrix. 
Given a set $\Lambda \subset [n]$ and a vector $\sigma_\Lambda \in\{0,1\}^\Lambda$, we define the two-sided influence matrix $\Psi_\mu^{\sigma_\Lambda}$ as
\[
\Psi_\mu^{\sigma_\Lambda}(i,j) := \prob{\xxi \sim \mu}{\xi_j =1|\xi_i =1\land \xi_
\ell = \sigma_\Lambda(\ell)\ \forall \ell \in \Lambda} - \prob{\xxi \sim \mu}{\xi_j=1| \xi_i=0 \land \xi_\ell = \sigma_\Lambda(\ell)\ \forall \ell \in \Lambda},
\]
if the conditioning is feasible (i.e. the event being conditioned on has non-zero probability) and $\Psi_\mu^{\sigma_\Lambda}(i,j) := 0$ otherwise.

We define also another variant of influence matrix that appears in \cite{L21a} and we call \emph{one-sided} influence matrix. 
Given a distribution $\mu$, a set of elements $\Lambda\subset [n]$ and two other indexes $i,j \in [n]\setminus \Lambda$ we define the influence of an element $i$ on another element $j$ when conditioning on $\Lambda$ as
$$ \calI_\mu^\Lambda(i\rightarrow j) := \prob{\xxi \sim \mu}{\xi_j =1|\xi_i =1\land \xi_\ell =1\ \forall \ell\in \Lambda} - \prob{\xxi \sim \mu}{\xi_j=1| \xi_\ell =1\ \forall \ell\in \Lambda} .$$
The matrix $\calI_\mu^{\Lambda}(i,j) := \calI_\mu^\Lambda(i\rightarrow j)$ if the conditioning on $\Lambda\cup\{i,j\}$ is feasible and $\calI_\mu^{\Lambda}(i,j) := 0$ otherwise is called the pairwise (one-sided) influence matrix.

With the definition of influence matrix, we can finally state the definition of {\iindependence}.
\begin{definition}\label{def:iindependence}
We say that a distribution $\mu$ is (one sided) \emph{\iindependent} with parameter $\Dinf$ if 
$$ \norm{\calI_\mu^\Lambda}_\infty = \max_{i\in [n]} \sum_{j\in [n]} |\calI_\mu^\Lambda(i\rightarrow j)| \leq \Dinf $$
for all subsets $\Lambda \subset[n]$.

We say that $\mu$ is \emph{\tsiindependent} with parameter $\Dinf$ if
$$ \norm{\Psi_\mu^{\sigma_\Lambda}}_\infty \leq \Dinf ,$$
for all $\Lambda \subset [n]$ and $\sigma_\Lambda \in\{0,1\}^\Lambda$.
\end{definition}

It is easy to see that {\tsiindependence} is stronger than {\iindependence} and in Appendix~\ref{sec:example_appendix} we show an example of $k$-homogeneous one-sided {\iindependent} distribution which is not {\tsiindependent}. Notice that two-sided {\iindependence} is a slightly stronger variant of the notion of spectral independence introduced in \cite{ALG21} that requires $\lambda_{\max}(\Psi_{\mu_{\sigma_\Lambda}}) \leq \eta$ for every $\sigma_\Lambda \in \{0,1\}^\Lambda, \Lambda \subset [n]$. 
Two-sided {\iindependence} is implied by SCP, but is not comparable with the notions of negative regression and negative association.
Although, to the best of our knowledge, nobody has named the property of {\iindependence} before, several papers \cite{L21a, AASV21a, ALG21, CLV21, CGSV21, FGYZ21, FGKP21b, CLV20, BCCPSV21} used {\tsiindependence} as a tractable way to prove spectral independence. Thus to our luck, we can furnish large number of distributions with bounded {\iindependence}.
In Section~\ref{sec:iindeplist}, we provide an overview of known families of distributions with  bounded {\iindependence}.

We show that linear matrix-valued functions of $k$-homogeneous {\iindependent} random variables satisfy matrix Chernoff bounds of the form presented by Tropp. 
We state our main theorem:
 \begin{theorem} \label{theorem:main_contribution}
 Let $\xi_1,\xi_2,\dots,\xi_n\in \{0,1\}$ be $n$ random variables with some joint distribution $\mu$ which is $k$-homogeneous and {\iindependent} with parameter $D$. Let $\YY_1,\dots,\YY_n\in \R^{d\times d}$ be a collection of symmetric matrices such that $0\preceq \YY_i\preceq R\II,\ \forall i\in [n]$ for some $R > 0$. Define $\mu_{\min}:= \lambda_{\min} \left(\expec{\xxi \sim \mu}{\sum_{i} \xi_i \YY_i}\right)$ and $\mu_{\max}:= \lambda_{\max} \left(\expec{\xxi \sim \mu}{\sum_{i} \xi_i \YY_i}\right)$. Then for any $\delta\in [0,1]$
 \[
   \prob{}{\lambda_{\min}\left( \sum_i \xi_i \YY_i \right) \leq (1-\delta) \mu_{\min}} \leq d \exp\left({-\frac{\delta^2 \mu_{\min}}{O(R D^2)} } \right) \quad \mbox{and}
 \]
 \[
   \prob{}{\lambda_{\max}\left( \sum_i \xi_i \YY_i \right) \geq (1+\delta) \mu_{\max}} \leq d \exp\left({-\frac{\delta^2\mu_{\max}}{O(R D^2)} }\right).
   \]
\end{theorem}

When specialized to distributions with the SCP, Theorem~\ref{theorem:main_contribution} answers positively the question posed by Kyng and Song \cite{KS18} on whether the $\log(k)$ factor in the exponent that appears in their matrix Chernoff bound for Strongly Rayleigh distributions can be removed. We obtain as corollary that $O(\log\abs{V})$ random spanning trees gives a spectral sparsifier of a (weighted, undirected) graph $G = (V,E)$ whp improving on the $O(\log^2\abs{V})$ bound given in \cite{KS18}. Our result thus matches the lower bound from \cite{KS18} (see their Theorem 1.8).
In Appendix~\ref{sec:tree_appendix}, we sketch a proof of the corollary and provide some relevant preliminaries.

\begin{corollary}
  \label{cor:spantreesparsify}
  Given as input a weighted graph $G$ with $n$ vertices and a parameter $\epsilon > 0$, let $T_1, T_2, \cdots, T_t$ denote $t$ independent inverse leverage score weighted random spanning trees.
We we choose $t = C\epsilon^{-2} \log n$ then with probability $1-1/n^{\Omega(C)}$,
\begin{align*}
(1-\epsilon)  \LL_G \preceq \frac{1}{t} \sum_{i=1}^t \LL_{T_i} \preceq (1+\epsilon) \LL_G.
\end{align*}
where $\LL_G$ is the Laplacian matrix of the graph $G$.
\end{corollary}

We also show that it is possible to relax the homogeneity assumption if we assume that the distribution is {\tsiindependent} as stated in the following corollary.
 \begin{corollary}\label{cor:homogenization}
 Let $\xi_1,\xi_2,\dots,\xi_n\in \{0,1\}$ be $n$ random variables with some joint distribution $\mu$ which is {\tsiindependent} with parameter $D$. Let $\YY_1,\dots,\YY_n\in \R^{d\times d}$ be a collection of symmetric matrices such that $0\preceq \YY_i\preceq R\II,\ \forall i\in [n]$ for some $R > 0$.
 Define $\mu_{\min}:= \lambda_{\min} \left(\expec{\xxi \sim \mu}{\sum_{i} \xi_i \YY_i}\right)$ and $\mu_{\max}:= \lambda_{\max} \left(\expec{\xxi \sim \mu}{\sum_{i} \xi_i \YY_i}\right)$. Then for any $\delta\in [0,1]$
 \[
   \prob{}{\lambda_{\min}\left( \sum_i \xi_i \YY_i - \expec{\xxi \sim \mu}{ \sum_i \xi_i \YY_i} \right) \leq \delta \mu_{\min}} \leq d \exp\left({-\frac{\delta^2 \mu_{\min}}{O(RD^2)} } \right) \quad \mbox{and}
 \]
 \[
   \prob{}{\lambda_{\max}\left( \sum_i \xi_i \YY_i - \expec{\xxi \sim \mu}{ \sum_i \xi_i \YY_i} \right) \geq \delta \mu_{\max}} \leq d \exp\left({-\frac{\delta^2\mu_{\max}}{O(RD^2)} }\right).
   \]
\end{corollary}

Compared to \cite{KS18}, we need several innovations to make our proofs go through. The first is to adopt different overall induction hypothesis which departs from the martingale framework of \cite{KS18} -- and crucially, we find a more strongly ``direction-aware'' induction hypothesis using trace matrix exponentials.
This accounts for our ability to shave off a $\log(k)$ factor compared to this work.
Our other innovations relate to the introduction of {\iindependence} which is a much more relaxed condition on a probability distribution than the Strongly Rayleigh property.
We show that {\iindependence} both provides a strong stability property for the distribution under conditioning.
We also show that {\iindependence} for $k$-homogeneous distributions is equivalent to a notion we call \emph{{\amindependence}} (Definition~\ref{def:amindependence}).
Average multiplicative independence allows us fine-grained direction-aware control over changes to the distribution, which is crucial to obtaining our matrix concentration bound.

Our result is much more broadly applicable than \cite{KS18}, as {\iindependence} has recently been established for a wide range of distributions. Furthermore, through a homogenization argument, we extend the result also to non-homogeneous distributions that fulfill a stronger notion of {\iindependence}. We remark that the bounds given in Theorem~\ref{theorem:main_contribution} do not explicitly depend on $k$ but only on $\expec{\xxi \sim \mu}{\sum_{i} \xi_i \YY_i}$. By virtue of this, our Chernoff-like bound generalizes naturally and extends without loss to the non-homogeneous case. In contrast, McDiarmid type bounds can incur a large loss when applied to non-homogeneous functions through homogenization.

Our proof of Theorem~\ref{theorem:main_contribution} is simple but substantially different from other proofs of concentration for dependent variables, and we believe it highlights a conceptually important point: Many recent works have shown McDiarmid-type bounds for concentration of $1$-Lipschitz functions using the Herbst argument, but these are fundamentally weaker than Chernoff-type bounds in many regimes.
The difference is especially stark in the matrix setting, where our bounds are exponentially stronger than similar for concentration of $1$-Lipschitz functions for related distributions.
We present a brief overview of the proof in Section~\ref{sec:proof_overview} and the complete proof in Section~\ref{sec:proof_complete}.

Paulin \cite{P14} considered the interdependence matrix $\AA$ of a distribution $\mu$ which is entry-wise bigger than the influence matrix $\calI_\mu^\Lambda$ for every $\Lambda \subset[n]$. Under the hypothesis that $\norm{\AA}_1 <1$ and $\norm{\AA}_\infty \leq 1$ the author proved a Chernoff bound statement. 
This result stops holding when one of the conditions $\norm{\AA}_1 <1$ or $\norm{\AA}_\infty \leq 1$ do not apply; our approach is instead much more broadly applicable since it adapts to the level of dependence of the random variables.

\paragraph{Further related work.}
A recent manuscript by Anari et al. \cite{AJKPV21} introduced a notion of \emph{entropic independence}, which is stronger than spectral independence and leads to stronger mixing time results when it applies -- including for Ising models in some regimes and for so-called fractionally log-concave polynomials.
Thus spectral independence is implied by both {\tsiindependence} and entropic independence, but the relation between the latter two is unclear.
A manuscript by Eldan and Shamir \cite{ES20} studies a notion of log-concave distributions over the hypercube that differs from the one discussed above and they use it to prove bounds on the variance of 1-Lipschitz functions of samples from the distribution.
Both the notion and the technique is substantially different -- and the implications for concentration are more limited as they rely on Chebyshev's inequality.

Expander Chernoff bounds \cite{G98} show that using an expander graph, we can take multiple ``pseudo-independent'' samples from a distribution over its vertices, by first sampling one vertex at random and then taking later samples using the trajectory of a random walk starting from this initial point. In particular, the work of \cite{G98} and later refinements showed that when we associate a scalar value to each vertex, the samples obtained using the random walk will exhibit concentration around the mean.
This provides a randomness-efficient way to obtain concentration.
This phenomenon was generalized to matrix-valued functions of the vertices by Garg et al. \cite{GLSS18a}.
Their result is closer to a direction-unaware matrix Chernoff than the classic direction-aware Chernoff bound of \cite{T12}.

\paragraph{Discussion and open questions.}
Recently, there has been a flurry of work on mixing times for random walks associated with various distributions over discrete, finite probability spaces.
This has led to the development of a broad array of new tools for understanding such distributions, including local spectral expansion, spectral independence, and entropic independence.
We view our work here as an early step toward an associated theory of concentration -- and we highlight that understanding the overall picture for concentration is likely to require moving beyond standard arguments from modified log-Sobolev or Poincar\'e inequalities.
%
%
While we have studied matrix Chernoff bounds, it is likely that many other matrix concentration bounds can be established from (two-sided) {\iindependence}. 

Finally, we wish to point out another basic open question.
We define a \emph{McDiarmid bound} for a $k$-homogeneous distribution $\calD$ over $\setof{0,1}^n$ as any statement along the lines of
\begin{quote}
  ``For any function $f : \setof{0,1}^n \to \R$ that is $1$-Lipschitz w.r.t. Hamming distance, if $\xxi \sim \calD$, then $\abs{f(\xxi) - \E{f(\xxi)}} \leq O(\sqrt{k \log(1/\delta)})$ with probability at least $1-\delta$.''
\end{quote}
We define a \emph{Chernoff bound} for a distribution $\calD$ over $\setof{0,1}^n$ as any statement along the lines of
\begin{quote}
  ``For any linear function linear $f(\xxi) = \sum_i c_i \xxi(i)$ with $c_i \in [0,1]$ for all $i$,
  \[
    \abs{f(\xxi) - \E{f(\xxi)}} \leq O(\sqrt{\E{f(\xxi)} \log(1/\delta)})
  \]
  with probability at least $1-\delta$, at provided $\log(1/\delta) \leq \E{f(\xxi)}$.''
\end{quote}
We could hope to combine the best features of both, and show $\abs{f(\xxi) - \E{f(\xxi)}} \leq O(\sqrt{\E{f(\xxi)} \log(1/\delta)})$ for a non-negative 1-Lipschitz function.
We are not aware of such a statement or a counterexample existing in the literature. 

\paragraph{Organization of the paper.}
In Section~\ref{sec:proof_overview}, we give an overview of the proof of Theorem~\ref{theorem:main_contribution}, our main theorem.
In Section~\ref{sec:iindeplist}, we list some known families of distributions with bounded {\iindependence}.
In Section~\ref{sec:preliminaries}, we state preliminaries and in Section~\ref{sec:proof_complete} we give the full proof of Theorem~\ref{theorem:main_contribution}.

\section{Overview of the proof}\label{sec:proof_overview}
Let $\mu$ be a $k$-homogeneous distribution. A natural way to prove Chernoff bound is to establish a bound on the moment generating function of the random variables. 
For matrix-valued random variables, we can deduce concentration from a bound on the expected trace of a matrix exponential of the variable.
We prove that for a certain constant $c\geq 1$, $$ \tr\left(\expec{\xxi\sim \mu}{e^{\theta \left( f(\xxi) - \expec{\xxi \sim \mu}{f(\xxi)} \right) }} \right) \leq \tr\left(e^{c\theta^2\expec{\xxi \sim \mu}{f(\xxi)}}\right) \quad \forall \theta \in \left(-\frac{1}{c},\frac{1}{c}\right),$$
for all matrix functions $f$ of the form $f(\xxi) =\sum_i\xi_i\YY_i$ with $0\preceq \YY_i\preceq\II$ for all $i$. Given this bound, the theorem follows from the standard Chernoff argument and the constant $c$ defines the quality of the bound. We show by induction over $k$ that a slightly stronger statement holds:
for every symmetric matrix $\HH$,
$$ \tr\left( \expec{\xxi \sim \mu}{e^{\HH + \theta \sum_{i}\xi_i \YY_i}} \right) \leq \tr \left( e^{\HH+  (\theta+c\theta^2)   \expec{\xxi \sim \mu}{\sum_{i}\xi_i \YY_i} }\right) .$$
Since the distribution $\mu$ is $k$-homogeneous, taking the expectation w.r.t. $\mu$ is equivalent to first drawing a variable $v\in[n]$ with a distribution proportional to the marginals of $\mu$ and then draw the other $(k-1)$ variables conditional on the first one.  Once we fix the first chosen variable $v$, the remaining variables are sampled based on a distribution which is  $(k-1)$-homogeneous and inherits {\iindependence} from $\mu$.  So, for each fixed $v$, we apply the inductive hypothesis to the distribution that samples the remaining $(k-1)$ variables conditional on $v$. In order to simplify the notation we introduce the quantity  $\ZZ_v := \expec{\xxi \sim\mu}{\sum_{i} \xi_i \YY_i|\xi_v =1} - \expec{\xxi\sim \mu}{\sum_{i} \xi_i \YY_i} $ that describes how much the expectation of the sum function changes when we know that our outcome must contain $v$.  Using the Golden-Thompson trace inequality and some linear algebra we then obtain:
$$ \tr\left( \expec{\xxi \sim \mu}{e^{\HH + \theta \sum_{i}\xi_i \YY_i}} \right) \leq \tr \left( e^{\HH+  (\theta+c\theta^2)   \expec{\xxi \sim \mu}{\sum_{i}\xi_i \YY_i} }\right) \norm{\expec{v\in [n]}{ e^{(\theta+c\theta^2) \ZZ_v - c\theta^2 \YY_v}}} $$
\par
In order to conclude our proof, it remains to show $\expec{v\in [n]}{ e^{(\theta+c\theta^2) \ZZ_v - c\theta^2 \YY_v}}  \preceq \II$ for some fixed $c$. In order to prove this, we show that there exist two constants $\Dinf$ and $\Dam$ such that: (a) $\ZZ_v\preceq \Dinf\II\ \forall v$ and (b) $\expec{v}{\ZZ_v} \preceq \Dam \expec{v}{\YY_v}$. The former condition is satisfied when the distribution $\mu$ is  {\iindependent} and the later is satisfied when $\mu$  has a property that we call {\amindependence}.
These properties require that the marginals do not change too much when we condition on a single variable, but quantify this in different ways. {\iindependence} requires that for every variable we may condition on,
we can bound the sum of the absolute values of changes in marginals of all other variables.
On the other hand, {\amindependence} imposes a multiplicative bound on the change in the marginal of each variable under conditioning, but averaged over the different variables we may condition on.

Surprisingly, we show that, in the case of $k$-homogeneous distributions, the {\iindependence} and {\amindependence} are equivalent. Hence if $\mu$ is {\iindependent} with parameter $\Dinf$ then it is also {\amindependent} with parameter $\Dam = \Dinf$. Finally we show that conditions (a) and (b) imply $\expec{v\in [n]}{ e^{(\theta+c\theta^2) \ZZ_v - c\theta^2 \YY_v}}  \preceq \II$ for $c=1/O(\Dinf \Dam)$.  Hence if $\mu$ is {\iindependent} with parameter $D$, it suffices to choose $c=1/O(D^2)$ that gives the factor $1/O(D^2)$ at the exponent in the bound. 


\section{Overview of {\iindependent} Distributions}
\label{sec:iindeplist}

\ifsodaversion
\else 

In this section, we give a list of some notable distributions with bounded {\iindependence}.

\subsection{Stochastic covering property}
Stochastic covering property \cite{PP14} is a form of negative dependence that is weaker than Strongly Rayleigh property.
\begin{definition} 
Let $(\xi_1,\dots,\xi_n)\in \{0,1\}^n$. We say that the distribution of $\xxi$ has the stochastic covering property (SCP) if for every set of indexes $\tau \subset [n]$ and for every index $v\in[n]$ the following holds. Let $\xxi'\in \{0,1\}^{n-|\tau|-1}$ be the distribution on entries $[n]\setminus (\tau \cup \{v\})$ of $\xxi$ conditional on $\xi_i=1$ for $i\in \tau$ and $\xi_v =0$. Let $\xxi''$ be the distribution of the same entries of $\xxi$ conditional on $\xi_i=1$ for $i\in \tau$ and $\xi_v=1$.
Then, there exists a coupling between $\xxi'$ and $\xxi''$ (i.e. a joint distribution of the two vectors), s.t. in every outcome of the coupling the value of $\xxi'$ can be obtained from the value of $\xxi''$ by either changing a single from 0 to 1 or by leaving all entries unchanged.
\end{definition}
\begin{proposition}\label{prop:SCP_iindep}
A $k$-homogeneous distribution that satisfy SCP is {\tsiindependent} with parameter 2.
\end{proposition}
The proof of the proposition can be found in Appendix~\ref{sec:proof_appendix}.

\subsection{Gibbs Distributions of Spin Systems}

Spin systems capture many combinatorial models of interest in statistical physics, the Gibbs distribution is a probability distribution defined over the collection of all the configurations of a spin system. 
\paragraph{The monomer-dimer model}
Given a graph $G = (V,E)$ a matching $M\subseteq E$ of $G$ is a set of edges without common vertexes. Let $G = (V,E)$ be a graph and $\lambda > 0$ be a real parameter, the Gibbs distribution $\mu$ for the monomer-dimer model with fugacity $\lambda$ is defined on the collection $\calM$ of all matchings of $G$ where $$\mu(M ) :=  \frac{\lambda^{|M|}}{Z}$$ and $Z = \sum_{M \in \calM} \lambda^{|M|}$.
\begin{theorem}[Theorem~6.1 from \cite{CLV21}]
Fix an integer $\Delta\geq 3$ and a real number $\lambda >0$. Then for every graph $G=(V,E)$ with maximum degree at most $\Delta$ and $|E|=m$ edges, for every $\Lambda\subseteq E$, for every feasible boundary condition $\tau:\Lambda\mapsto \{0,1\}$, the Gibbs distribution of the monomer-dimer model with fugacity $\lambda$ is {\tsiindependent} with parameter $1+\min\{2\lambda\Delta,s\sqrt{1+\lambda\Delta}\}$
\end{theorem}

\paragraph{Ising/Potts model}
Consider a graph $G=(V\cup \partial V,E\cup \partial E)$ of degree at most $\Delta$, where $\partial E$ contains the edges between $V$ and $\partial V$. Let $q\geq 2$ be a positive integer, $\xi \in [q]^{\partial V}$ a vector of boundary conditions, $\AA_{xy}\in \R^{q\times q}$ for $\{x,y\} \in E \cup \partial E$ a collection of matrices representing the nearest neighbour interactions, and $\BB_x\in \R_{>0}^q$ for $x\in V$ a collection of vectors representing the external fields. A vector $\sigma \in [q]^{V}$ is called a configuration of the model and corresponds to an assignment of a label form $[q]$ for every vertex in $V$ . The Gibbs distribution the $q$-spin system is defined as the distribution $\mu$ over all the possible configurations $\sigma \in [q]^{V}$ such that
$$ \mu(\sigma) :=\frac{1}{Z_G} \prod_{\{x,y\}\in E} A_{xy}(\sigma_x,\sigma_y) \prod_{\{x,y\} \in \partial E, x\in V,y\in \partial V} A_{xy}(\sigma_x,\xi_y)\prod_{x\in V}B_x(\sigma_x), $$
where $Z_G$ is a normalization term such that $\sum_{\sigma\in [q]^V}\mu(\sigma)=1$.\\

The Gibbs distribution is not defined on $\setof{0,1}^n$ but we can encode a configuration $\sigma \in [q]^V$ with a binary vector $\xxi \in \setof{0,1}^{|V| q}$ such that $\xi(v,i) =1 \iff \sigma(v)=i$, for all $v\in V$ and $i\in [q]$. 
\begin{theorem}[Theorem 4.13\cite{BCCPSV21}]
Let $\beta \in \R$, $\Delta \geq 3$, $q\geq 2$, and $\hh \in \R^q$. The Ising/Potts model is a particular case of $q$-state spin system where $A_{xy}(i,j)= \exp(\beta \delta_{ij})$ and $B_x(i) =\exp(h(i))$. For the Ising/Potts model, if $\beta <\max\{\frac{2}{\Delta},\frac{1}{\Delta}\ln(\frac{q-1}{\Delta})\}$ then the Gibbs distribution $\mu$ is {\iindependent} with parameter $\eta(\beta,\Delta)$. 
\end{theorem}

\paragraph{List colouring}
Consider a graph $G=(V,E)$ of degree at most $\Delta$. Let $q\geq \Delta+2$ be a positive integer and suppose we are given a collection of lists $(L(v))_{v\in V}$ one for each vertex of $G$ and that $L(v)\subseteq [q],\ \forall v \in V$. We call an assignment for each vertex of an element of in its list $\sigma \in \prod_{v\in V} L(v)$ a list-coloring of $G$. We say a list-coloring $\sigma$ is proper if $\sigma(u)\neq\sigma(v)\ \forall (u,v)\in E$ i.e. all neighbours are assigned to different elements. \\
Also here we use the encoding we defined for the Ising/Potts model to transform a colouring into a $\setof{0,1}^n$ vector.

\begin{theorem}[See Theorem 1.1, Theorem 1.3, and Lemma~5.2 \cite{L20}] 
Let $(G, \calL)$ be a list-coloring instance where $G = (V, E)$ is a graph of maximum degree $\Delta \leq O(1)$ and $\calL = (L(v))_{v\in V}$ is a collection of color lists of maximum length $q$. Then for some absolute constant $\epsilon \approx 10^{-5}$, if $q \geq(\frac{11}{6} -\epsilon)\Delta$ , then the uniform distribution over proper list-colorings for $(G, \calL)$ is {\iindependent} with parameter $O(1)$.
\end{theorem}

\begin{theorem}[See proof of Theorem 9~\cite{CGSV21}, see also Lemma~6.1~\cite{FGYZ21}]
Let $\epsilon >0$, and suppose that $(G,\calL)$ is a list colouring instance where $G = (V, E)$ is a triangle-free graph of maximum degree $\Delta$ and $\calL = (L(v))_{v\in V}$ is a collection of color lists of maximum length $q\geq (1+\epsilon)\alpha^* \Delta +1$, $\alpha^*\approx 1.763$. Then the uniform distribution over proper list-colorings for $(G, \calL)$ is {\iindependent} with parameter $64(\frac{1}{\epsilon} +1)^2\frac{\Delta}{q}+1$.
\end{theorem}

\subsection{Stable distributions}
Chen et al. \cite{CLV21b} generalize the results of Alimohammadi et al \cite{AASV21a} which formalize a connection between stability of polynomials and spectral independence of certain probability distributions. By inspection of the proof, we notice that in order to prove spectral independence, in both papers, the authors show a stronger result, namely that the distributions of interest are {\iindependent}. We summarise in this section the results obtained by Chen et al.
\par
Let $q \geq 2$ be an integer, $V$ be a finite set of vertices, and $Q=\setof{0,1,\dots,q-1}$. A spin assignment $\sigma : V \mapsto Q$ is called a configuration. Denote with $\Omega = Q^{V}$ be the set of all the configurations and let $w: \Omega \mapsto \R_{\geq 0}$ be a nonnegative weight function that is not identically zero. A configuration $\sigma \in \Omega$ is called feasible if $w(\sigma)>0$. For $\Lambda \subseteq V$ define the set of of pinnings on $\Lambda$ by 
$$ \Omega_{\Lambda} = \setof{ \tau \in Q^\Lambda: \exists\text{ feasible }\sigma \in \Omega\text{ s.t. }\sigma_\Lambda = \tau} $$
Let $\calT = \union_{\Lambda \subseteq V} \Omega_{\Lambda}$ be the collection of all pinnings. For $\tau \in \Omega_\Lambda$, let $V^\tau = V \setminus \Lambda$ denote the set of unpinned vertices. For $v\in V^\tau$ let 
$$ \Omega_v^\tau = \setof{ k\in Q: \exists \text{ feasible } \sigma \in \Omega\text{ s.t. }\sigma_\Lambda = \tau\text{ and }\sigma_v = k}$$
Given a complex function $\llambda$ that associates a complex number $\lambda_{v,\sigma_v}$ to each pair $(v, \sigma_v) $ such that $v\in V^\tau$, $\sigma_v\in \Omega_v^\tau$, and $\sigma_n \neq 0$, the conditional partition function under $\tau$ is 
$$ Z_w^\tau := \sum_{\sigma \in \Omega: \sigma_ \Lambda = \tau} w(\sigma) \llambda^{\sigma_U}, \text{ where } \llambda^{\sigma_U} = \prod_{v\in U: \sigma_v \neq 0} \lambda_{v,\sigma_v}.$$
Then, define the Gibbs distribution as 
$$ \mu^{\tau}(\sigma) = \mu(\sigma|\sigma_\Lambda =\tau) = \frac{w(\sigma) \llambda^{\sigma_U}}{Z_w^\tau(\llambda)}, \quad \forall \sigma \in \Omega \text{ s.t. } \sigma_{\Lambda} = \tau.$$
Finally, we define when a polynomial is stable.

\begin{definition}
For an integer $n\geq 1$ and $\calK\subset \C^n$, we say a multivariate polynomial $P\in \C[z_1,\dots,z_n]$ is $\calK$-stable if $P(z_1,\dots,z_n) \neq 0$ whenever $(z_1,\dots,z_n) \in \calK$. In particular if $\calK= \prod_{\ell=1}^{n} \Gamma$ for some $\Gamma \subseteq \C$, then we say $P$ is $\Gamma$-stable.
\end{definition}

Note that following the notation in \cite{CLV21b}, the Gibbs distribution is not defined on $\setof{0,1}^n$ but we can encode a configuration $\sigma \in Q^V$ with a binary vector $\xxi \in \setof{0,1}^{|V| q}$ such that $\xi(v,i) =1 \iff \sigma(v)=i$, for all $v\in V$ and $i\in [q]$. Note that the partition function is not changed by this change of encoding.

The following two theorems show that there is a relation between the stability of the partition function and the {\iindependence} of the Gibbs distribution.

\begin{theorem}[From proof of Theorem 7 \cite{CLV21b}]
Let $\Gamma \subset \C$ be a non-empty open connected region such that $\Gamma$ is unbounded and $0$ belongs to the closure of $\Gamma$. If the multivariate partition function $Z_w^\tau$ is $\Gamma$-stable, then for any $\lambda\in \R^+ \cap \Gamma$ the Gibbs distribution $\mu =\mu_{w,\lambda}$ with the uniform external field $\lambda$ is {\iindependent} with constant
$$ \frac{8}{\delta} $$
where  $\delta = \frac{1}{\lambda} \dist(\lambda,\partial \Gamma)$.
\end{theorem}

If we restrict to the case $q=2$, then we obtain that $\Omega = \setof{0,1}^{|V|}$ is the binary hypercube and the conditional partition $ Z_w^\tau $ is the generating polynomial of the distribution $w$. 
\begin{corollary}
Let $w:2^{[n]}\mapsto \R_{\geq0}$ be a distribution and $\Gamma \subset \C$ be a non-empty open connected region such that $\Gamma$ is unbounded and $0$ belongs to the closure of $\gamma$. If the generating multi-affine polynomial $Z_w \in \R[z_1,\dots,z_n]$ is $\Gamma$-stable,  then for any $\lambda\in \R^+ \cap \Gamma$ the Gibbs distribution $\mu =\mu_{w,\lambda}$ with the uniform external field $\lambda$ is {\iindependent} with constant
$$ \frac{8}{\delta} $$
where  $\delta = \frac{1}{\lambda} \dist(\lambda,\partial \Gamma)$.
\end{corollary}

Note that this result generalizes the result obtained in \cite{AASV21a} that restricts to the case in which $\Gamma$ is a sector.

The next theorem states that in the case in which the region $\Gamma$ is not unbounded, we have a similar result but we need to assume that the conditional partition function $Z_w^\tau$ is stable for every pinning $\tau \in \calT$. 
Let $\mathcal{P} =\setof{(v,k)\in V\times Q:v\in V^\tau,k \in \Omega_v^\tau}$, define the marginal bound for a weight function $\mu$ as $$ b = \min_{\tau \in \calT, (v,k) \in \mathcal{P}^\tau} \mu^{\tau} (\sigma_v=k).$$

\begin{theorem}[From the proof of Theorem 8 \cite{CLV21b}]
Let $\lambda^* \in \R^+$ and let $\Gamma \subset \C$ be a non-empty open connected region such that $(0, \lambda^* )\subseteq \Gamma$ (respectively, $(\lambda^*,\infty) \subseteq \Gamma$). If for every pinning $\tau \in \calT$ the multivariate conditional partition function $Z_w^\tau$ is $\Gamma$-stable, then for any $\lambda\in (0,\lambda^*)$ (respectively, $ \lambda\in (\lambda^*, \infty )$ ) the Gibbs distribution $\mu =\mu_{w,\lambda}$ with the uniform external field $\lambda$ is {\iindependent} with constant
$$ \frac{8}{\delta} \min \setof{\frac{1-b}{b}, \frac{\lambda}{b(\lambda^* -\lambda)} +1} $$
$$ (\text{respectively} , \frac{8}{\delta} \min \setof{\frac{1-b}{b}, \frac{\lambda^*}{b(\lambda -\lambda^*)} +1}  )$$
where $b$  is the marginal bound for $\mu$ and $\delta = \frac{1}{\lambda} \dist(\lambda,\partial \Gamma)$.
\end{theorem}



\fi

\section{Preliminaries}\label{sec:preliminaries}

Let $\xi_1,\xi_2,\dots,\xi_n\in \{0,1\}$ be $n$ random variables with some joint distribution $\mu$.  We think about an outcome from the distribution $\mu$ as both a vector $\xxi \in\{0,1\}^n$ and a set of indices $\sigma\subset [n]$ such that $v\in \sigma$ if and only if $\xi_v =1$.
 Let $\pp\in [0,1]^n$ be the vector of the marginal probabilities when sampling from the distribution $\mu$ i.e. $$ \pp := \expec{\xxi \sim \mu}{\xxi}.$$
We say that $\mu$ is $k$ homogeneous when every outcome contains exactly $k$ ones i.e.  $\norm{\xxi}_1 =k$ with probability $1$. In this paper we are mainly interested in homogeneous distributions. \par
For $k$-homogeneous distributions, we can sample an outcome in the following way: we first pick a single variable based on a certain distribution and then we pick the remaining $(k-1)$ variables based on the first variable that we sampled. Formally, we define the probability distribution $\nu$ over the elements of $[n]$ as
$$ \nu(v) := \frac{1}{k} \sum_{\sigma \subseteq [n], v\in \sigma} \mu(\sigma)= \frac{p(v)}{k}$$
where $p(v)$ denotes the $v$-th entry of the vector $\pp$. 
Note that $\sum_{v\in [n]} \nu(v) =1$.\\
For every $v\in[n]$ such that $p(v)>0$, the probability distribution $\mu_v$ of $\mu$ conditioned on $\xi_v=1$ is defined for a every set $\tau\subseteq [n]\setminus\{v\}$ of size $|\tau|=k-1$ as  
$$ \mu_v(\tau) := \frac{\mu(\{v\}\cup \tau)}{\sum_{\sigma \subseteq [n], v\in \sigma}\mu(\sigma)}.$$
Analogously, for every $v\in[n]$ such that $p(v)<1$, the probability distribution $\mu_{\setminus v}$ of $\mu$ conditioned on $\xi_v=0$ is defined for a every set $\tau\subseteq [n]\setminus\{v\}$ of size $|\tau|=k$ as  
$$ \mu_{\setminus v}(\tau) := \frac{\mu(\tau)}{\sum_{\sigma \subseteq [n], v\notin \sigma}\mu(\sigma)}.$$
If $\sum_{\sigma \subseteq [n], v\in \sigma}\mu(\sigma)=0$ (resp. $\sum_{\sigma \subseteq [n], v\in \sigma}\mu(\sigma)=1$) then say that the conditioning on $\xi_v=1$ or $v$ being always present (resp. $\xi_v = 0$ or $v$ being always absent) is not feasible and $\mu_v$ (resp. $\mu_{\setminus v}$) is not defined. 
We define the vector of the marginals of $\mu_v$ and $\mu_{\setminus v}$ respectively as
$$\pp_{v}:= \expec{\xxi\sim \mu}{\xxi|\xi_v =1} \quad \text{ and } \quad \pp_{\setminus v}:= \expec{\xxi\sim \mu}{\xxi|\xi_v =0} .$$ 

Finally, we extend the conditioning to more than one variable. For a set $\Lambda \subset [n]$ and a vector $\sigma_{\Lambda} \in \{0,1\}^\Lambda$, we define the distribution $\mu_{\sigma_\Lambda}$ by conditioning on $\xi_i = \sigma_\Lambda(i)$ for every $i\in \Lambda$. If $\prob{\xxi\sim \mu}{\xi_i = \sigma_\Lambda(i),\ \forall i\in \sigma} =0$ then we say that the conditioning is not feasible and $\mu_{\sigma_\Lambda}$ is not defined.

\paragraph{Linear algebra}
We denote vectors with bold lower case letters and matrices with bold capital letters, in particular $\II$ denotes the identity matrix. For a matrix $\AA\in \R^{d\times d}$, we denote the maximum and minimum eigenvalues of $\AA$ respectively with $\lambda_{\max}(\AA)$ and $\lambda_{\min}(\AA)$. \\
Given an scalar function $f:\R \mapsto \R$ and a p.s.d. matrix $\AA$ with spectral decomposition $\VV \LLambda \VV^\trp $ we define $$f(\AA):= \VV \diag_i(\lambda_i(\AA)) \VV^\trp .$$
The $\ell_\infty$ norm of a matrix $\AA \in \R^{d\times d}$ is defined as
\[
\norm{\AA}_\infty := \max_i \sum_{j} |A(i,j)| = \max_{\xx\neq \veczero} \frac{\norm{\AA \xx}_\infty}{\norm{\xx}_\infty}.
\]
The trace of a square matrix $\AA$, denoted with $\tr(\AA)$ is defined to be the sum of elements on the main diagonal of $\AA$ and is also equal to the sum of its eigenvalues counted with multiplicities.
The trace is a linear function in the sense that $\tr(\alpha \AA) = \alpha \tr(\AA)$ and $\tr(\AA+ \BB)= \tr(\AA) +\tr(\BB)$ for every $\AA,\BB \in \R^{d\times d}$ and $\alpha \in \R$. Furthermore the following theorem holds.
\begin{theorem}[Golden-Thompson]
Let $\AA, \BB \in \R^{d\times d}$ be two symmetric matrices, then
$$ \tr\left(e^{\AA+\BB} \right) \leq \tr\left(e^{\AA} e^{\BB} \right) $$
\end{theorem}
The following facts will be useful later, the proof of the facts can be found in  Appendix~\ref{sec:proof_appendix}.

\begin{fact}\label{fact:matrix_sqr}
Let $\AA, \BB \in \R^{d\times d}$ be two symmetric matrices, then
$$ \AA\BB+\BB\AA \preceq \AA^2+ \BB^2$$
\end{fact}

\begin{fact}\label{fact:matrix_exp}
Let $\AA, \BB \in \R^{d\times d}$ be two symmetric matrices, such that $\AA\preceq \II$ and $\BB$ is p.s.d. then
$$ e^{\AA - \BB} \preceq \II + \AA -  \BB +2 \AA^2+2\BB^2$$
\end{fact}

\paragraph{{\iindependence} and {\amindependence}}
The notion of {\iindependence} has already been defined in the previous section (see Definition~\ref{def:iindependence}).
We introduce here another property that we call {\amindependence} and that will be crucial for the proof of the main theorem.
\begin{definition}\label{def:amindependence}
We say that a $k$-homogeneous distribution $\mu$ is {\amindependent} with parameter $\Dam$ if  
$$ \expec{v\sim \nu}{|\calI_\mu^\Lambda(v\rightarrow u)|} \leq \Dam \expec{\nu}{\xi_u} , \quad \forall u\in [n]$$ 
for every subset $\Lambda\subset [n]$.
\end{definition}

Coming back to the previous notation, {\iindependence} and {\amindependence} are equivalent of having that for all the distributions $\mu_\tau$ for any $\Lambda \subset [n]$ and $\tau = \{1\}^\Lambda$, respectively,
$$ \norm{\pp_v - \pp}_1 \leq \Dinf, \quad \forall v\in [n] $$ 
and
$$ \expec{v\sim \nu}{|\pp_v - \pp|} \leq \frac{\Dam}{k} \pp .$$ 
Surprisingly, for $k$-homogeneous distributions, the two notions are equivalent as shown in the following lemma, which we prove in the next section.
\begin{lemma} \label{lemma:ii-rmi_preliminaries}
A $k$-homogeneous distribution $\mu$ is {\iindependent} with parameter $D$ if and only if it is {\amindependent} with parameter $D$.
\end{lemma}

We introduce a formal tool that turns a non-homogeneous distribution into an $n$-homogeneous distribution. We will use this tool in order to extend our result to non-homogeneous distributions.
\begin{definition}
Given a distribution $\mu$ over $[n]$. The \emph{homogenization} of $\mu$ is a distribution $\muhom$ over $[2n]$ such that:
\[
\begin{cases}
\muhom(\sigma\cup\setof{i+n|i\in [n]\setminus \sigma}) := \mu(\sigma) & \forall \sigma \subseteq [n],\\
\muhom(\sigma) := 0 & \mbox{otherwise.}
\end{cases}
\]
\end{definition}

In words, $\muhom$ is obtained from $\mu$ by completing every outcome $\sigma\subseteq [n]$ with $(i+n)$ for all $i$ not in $\sigma$. The distribution $\muhom$ is hence always $n$-homogeneous and furthermore, the first $n$ entries of a random vector $\xxi_\text{hom}\in\setof{0,1}^{2n},\xxi_{\text{hom}}\sim \muhom$ are distributed as the entries of a vector $\xxi\in \setof{0,1}^n, \xxi \sim\mu$. 

\section{Proof of the Chernoff Bound}\label{sec:proof_complete}
In this section, we are going to prove Theorem~\ref{theorem:main_contribution} which is restated below. The main part of the proof is to establish a bound on the moment generating function given in Lemma~\ref{lemma:first_part}. Given Lemma~\ref{lemma:first_part}, the proof of the theorem follows from the standard Chernoff argument and is reported later in this section. In order to simplify the exposition of the proof of Lemma~\ref{lemma:first_part} we use the helper Lemma~\ref{lemma:second_part} that we are going to prove immediately after. We conclude the section with the proof of Corollary~\ref{cor:homogenization}.

Below, we restate Theorem~\ref{theorem:main_contribution}, however, we omit the scalar parameter $R$. The full Theorem~\ref{theorem:main_contribution} follows directly by rescaling the variables.

\begin{theorem}[Theorem~\ref{theorem:main_contribution} restated.] \label{thm:main_theorem}
Suppose $(\xi_1,  \xi_2,\dots, \xi_n) \in \{0, 1\}^n$ is a random vector of $\{0, 1\}$ variables whose distribution $\mu$ is $k$-homogeneous and {\iindependent} with parameter $D$.  Let $\YY_1, \YY_2, \dots , \YY_m\in \R^{d\times d}$ be a collection of symmetric positive semidefinite matrices such that $\norm{\YY_i} \leq 1, \forall i$ and $\mu_{\min} \II \preceq \expec{\xxi \sim \mu}{\sum_{i} \xi_i \YY_i} \preceq \mu_{\max} \II$. Then $\forall\ 0\leq \delta\leq 1$, 
$$ \prob{}{\lambda_{\min}\left( \sum_i \xi_i \YY_i \right) \leq (1-\delta) \mu_{\min}} \leq d e^{-\frac{\delta^2}{20D^2} \mu_{\min}}$$
$$ \prob{}{\lambda_{\max}\left( \sum_i \xi_i \YY_i \right) \geq (1+\delta) \mu_{\max}} \leq d e^{-\frac{\delta^2}{20D^2} \mu_{\max}}$$
\end{theorem}

\par In the remaining of the section, we suppose that we are given a distribution $\mu$ that is $k$-homogeneous, {\iindependent} w.p. $\Dinf$, and {\amindependent} w.p. $\Dam$.
\begin{lemma}\label{lemma:first_part}
Suppose that $\mu$ is $k$-homogeneous, {\iindependent} w.p. $\Dinf$, and {\amindependent} w.p. $\Dam$.
 There exists a choice of  $c=\Theta(\Dinf \Dam)$ such that for every $\theta \in[-1/(2c),1/(2c)]$ and for every symmetric matrix $\HH$,
$$ \tr\left( \expec{\xxi \sim \mu}{e^{\HH + \theta \sum_{i}\xi_i \YY_i}} \right) \leq \tr \left( e^{\HH+  (\theta+c\theta^2)   \expec{\xxi \sim \mu}{\sum_{i}\xi_i \YY_i} }\right) $$
\end{lemma}

For the sake of the clarity of the exposition we prove Lemma~\ref{lemma:first_part} taking for granted the following Lemma~\ref{lemma:second_part} that we are going to prove later.
\begin{lemma}\label{lemma:second_part}
Suppose that $\mu$ is $k$-homogeneous, {\iindependent} w.p. $\Dinf$, and {\amindependent} w.p. $\Dam$. Suppose also that $k\geq1$. For a fixed $v\in [n]$ define the quantity $\ZZ_v := \expec{\xxi\sim \mu}{\sum_{i} \xi_i \YY_i} - \expec{\xxi \sim\mu_v}{\sum_{i} \xi_i \YY_i} $ that measures how much the conditioning on $v$ changes the expectation of the sum. Then if $c\geq 5\Dinf \Dam$, for all $\theta \in[-1/(2c),1/(2c)]$,
$$ \expec{v\sim \nu}{ e^{(\theta+c\theta^2) \ZZ_v - c\theta^2 \YY_v}} \preceq \II.$$
\end{lemma}

\begin{proof}[Proof of Lemma~\ref{lemma:first_part}]
We prove the statement by induction on the number $k$ of elements in each outcomes. If the only possible outcome is the empty set i.e. $\xxi =0$ deterministically then 
$$ \tr\left( \expec{\xxi \sim \mu}{e^{\HH + \theta \sum_{i}\xi_i \YY_i}} \right) = \tr\left( e^{\HH} \right) =\tr \left( e^{\HH+  (\theta+c\theta^2)   \expec{\xxi \sim \mu}{\sum_{i}\xi_i \YY_i} }\right) .$$ Hence the statement is trivially true. \\

Suppose that $\mu$ is $k$-homogeneous with $k\geq 1$, we can think about sampling an outcome from $\mu$ by first sampling one variable $v\sim \nu$ and then remaining $(k-1)$ variables with the distribution $\mu_v$. By the law of total probability we can write:
\begin{equation}\label{eq:1}
 \tr\left( \expec{\xxi \sim \mu}{e^{\HH +\theta \sum_{i} \xi_i \YY_i}} \right) = \expec{v \sim \nu}{ \tr \left(  \expec{\xxi \sim \mu_v}{e^{\HH + \theta\sum_{i} \xi_ i\YY_i}} \right)} 
\end{equation}
The outcomes of the distribution $\mu_v$ always contain $v$ and exactly $k-1$ other variables. Furthermore, since $\mu$ satisfy {\iindependence}  and {\amindependence}, then also $\mu_v$ for every $v\in [n]$ must satisfy the {\iindependence} and {\amindependence} with the same parameters. As a consequence, the restriction of the distribution $\mu_v$ to all the variables except $v$ is $(k-1)$-homogeneous and satisfy the hypothesis of the lemma.   Hence we can apply the inductive hypothesis and we get
\begin{equation}\label{eq:2}
\tr \left(  \expec{\xxi \sim \mu_v}{e^{\HH + \theta \YY_v + \theta\sum_{i\neq v} \xi_i \YY_i}} \right) \leq \tr \left(  e^{\HH + \theta \YY_v + (\theta+c\theta^2) \expec{\xxi \sim \mu_v}{\sum_{i\neq v}\xi_i \YY_i}} \right). 
\end{equation}
In order to simplify the notation we introduce the quantity $\ZZ_v := \expec{\xxi \sim\mu_v}{\sum_{i} \xi_i \YY_i} - \expec{\xxi\sim \mu}{\sum_{i} \xi_i \YY_i} $.
We substitute (\ref{eq:2}) into (\ref{eq:1}).
\begin{eqnarray*}
 \tr\left( \expec{\xxi \sim \mu}{e^{\HH +\theta \sum_{i} \xi_i \YY_i}} \right) &\leq &\expec{v \sim\nu}{ \tr \left(  e^{\HH + (\theta+c\theta^2) \expec{\xxi \sim\mu_v}{\sum_{i}\xi_i \YY_i} - c \theta^2 \YY_v } \right)} \\
 &= & \expec{v\sim \nu}{ \tr \left(  e^{\HH + (\theta+c\theta^2) \expec{\xxi\sim \mu}{\sum_{i} \xi_i \YY_i} + (\theta+c\theta^2) \ZZ_v -c\theta^2 \YY_v}  \right)}  \\
 &\leq & \expec{v\sim \nu}{ \tr \left(  e^{\HH + (\theta+c\theta^2) \expec{\xxi\sim \mu}{\sum_{i} \xi_i \YY_i}} e^{ (\theta+c\theta^2) \ZZ_v -c\theta^2 \YY_v}  \right)} \quad\mbox{(Golden-Thompson ineq.)}  \\
  &= &   \tr \left(  e^{\HH + (\theta+c\theta^2) \expec{\xxi \sim \mu}{\sum_{i}\xi_i \YY_i}} \expec{v\sim\nu}{ e^{ (\theta+c\theta^2) \ZZ_v - c\theta^2 \YY_v}}  \right)  \\
& \leq & \tr \left(  e^{\HH + (\theta+c\theta^2) \expec{\xxi \sim \mu}{\sum_{i}\xi_i \YY_i}}\right) \norm{\expec{v\sim \nu}{ e^{(\theta+c\theta^2) \ZZ_v - c\theta^2 \YY_v}}} \\
& \leq & \tr \left(  e^{\HH + (\theta+c\theta^2) \expec{\xxi \sim \mu}{\sum_{i}\xi_i \YY_i}}\right) \quad \mbox{(Lemma~\ref{lemma:second_part})}
\end{eqnarray*}

\end{proof}

It remains to prove Lemma~\ref{lemma:second_part}. We first provide the following claims where we exploit {\iindependence} and {\amindependence} of $\mu$. 
\begin{claim}\label{claim:bound_Zv}
Under the hypothesis of Lemma~\ref{lemma:second_part}.
\begin{enumerate}[label=(\roman*)]
\item 
$$ \ZZ_v \preceq \Dinf \II $$
\item $$ \expec{v\sim \nu}{\ZZ_v^2} \preceq \Dinf \Dam \expec{v\sim \nu}{\YY_v} $$
\end{enumerate}
\end{claim}

\begin{proof}
\begin{enumerate}[label=(\roman*)]
\item
\begin{eqnarray*}
\ZZ_v & = & \sum_i \YY_i (p_{v}(i) - p(i)) \\
 & \preceq &  \sum_i \YY_i |p_{v}(i) - p(i)| \\
 & \preceq & \norm{\pp_v -\pp}_1 \II \\ 
 & \preceq & \Dinf \II \quad \mbox{(\iindependence)}
\end{eqnarray*}
\item
Let $\delta_v(i) := \begin{cases} +1, & p_v(i) - p(i) \geq 0 \\ -1,&  p_v(i) - p(i) < 0 \end{cases}$; so that $p_v(i) - p(i) =  \delta_v(i) |p_v(i) - p(i)|$.
\begin{eqnarray*}
\ZZ_v^2 & = & \left( \sum_i \YY_i \left(p_v(i) -p(i)\right) \right)^2\\
& = & \sum_i (p_v(i)-p(i))^2 \YY_i^2 + \sum_i \sum_{j<i} |p_v(i)-p(i)| |p_v(j)-p(j)| \left(\delta_v(i) \YY_i \delta_v(i) \YY_j + \delta_v(j)\YY_j \delta_v(i)\YY_i\right)\\
& \preceq & \sum_i (p_v(i)-p(i))^2 \YY_i^2 + \sum_i \sum_{j<i} |p_v(i)-p(i)| |p_v(j)-p(j)|((\delta_v(i) \YY_i)^2 + (\delta_v(j) \YY_j)^2 \quad \mbox{(Fact~\ref{fact:matrix_sqr})} \\
& \preceq & \sum_i |p_v(i)-p(i)| \sum_j |p_v(j)-p(j)| \YY_j    \quad \mbox{(since $\YY_i^2\preceq \YY_i$)}  \\
& \preceq & \Dinf \sum_j |p_v(j)-p(j)| \YY_j  \quad \mbox{(\iindependence)}
\end{eqnarray*}
\begin{eqnarray*}
 \expec{v\sim \nu}{\ZZ_v^2} & \preceq & \Dinf \expec{v\sim \nu}{\sum_i \YY_i |p_{v}(i) - p(i)|} \\
 & = & \Dinf \sum_v \nu(v) \sum_i \YY_i |p_{v}(i) - p(i) | \\
 & = & \Dinf \sum_i \YY_i  \sum_v \nu(v) |p_v(i) - p(i) |  \quad \mbox{(exchange the sums )} \\
 & \preceq & \Dinf \Dam  \sum_i \YY_i \frac{p(i)}{k} \quad \mbox{(\amindependence)} \\
 & = & \Dinf \Dam \expec{i\sim \nu}{\YY_i} 
\end{eqnarray*}
\end{enumerate}
\end{proof}

We can now proceed with the proof of Lemma~\ref{lemma:second_part}.
\begin{proof}[Proof of Lemma~\ref{lemma:second_part}.]
Since for each $v\in[n]$, 
\begin{eqnarray*}
(\theta+c\theta^2) \ZZ_v - c\theta^2 \YY_v & \preceq & (\theta + c \theta^2) \ZZ_v  \\
& \preceq & \frac{3}{4c} \ZZ_v \quad \mbox{($c\theta \leq 1/2$)} \\
& \preceq & \II \quad \mbox{(Claim~\ref{claim:bound_Zv} i, for $c\geq \Dinf$),}
\end{eqnarray*}
we can apply Fact~\ref{fact:matrix_exp} to $e^{(\theta+c\theta^2) \ZZ_v - c\theta^2 \YY_v} $.\
\begin{eqnarray*}
\expec{v \sim \nu}{ e^{(\theta+c\theta^2) \ZZ_v - c\theta^2 \YY_v}} & \preceq & \expec{v\sim \nu}{ \II + (\theta+c\theta^2) \ZZ_v - c\theta^2 \YY_v + 2(\theta+c\theta^2)^2 \ZZ_v^2 + 2 (c\theta^2)^2 \YY_v^2 } \\
& \preceq &  \II  - c\theta^2 \expec{v\sim \nu}{\YY_v} + \frac{9}{2} \theta^2 \expec{v\sim \nu}{\ZZ_v^2} + \frac{1}{2} \theta^2 \expec{v\sim \nu}{\YY_v^2 } \quad \mbox{($c\theta\leq 1/2$, $\expec{v}{\ZZ_v}=0$)} \\
& \preceq & \II  - c \theta^2 \expec{v\sim \nu}{\YY_v} + \frac{9}{2} \theta^2 \Dinf \Dam \expec{v\sim \nu}{\YY_v} + \frac{1}{2} \theta^2 \expec{v\sim \nu}{\YY_v} \quad  \mbox{(Claim~\ref{claim:bound_Zv} ii)} \\
& = &  \II  - \theta^2 \expec{v\sim \nu}{\YY_v} \left(c - \frac{9}{2} \Dinf \Dam - \frac{1}{2} \right) \\\
& \preceq & \II \quad \mbox{(for $c\geq 5 \Dinf\Dam$)}
\end{eqnarray*}
\end{proof}
Before giving the proof of Theorem~\ref{thm:main_theorem}, we give the proof of Lemma~\ref{lemma:ii-rmi_preliminaries} that we restate here for convenience. 
\begin{lemma}[Lemma~\ref{lemma:ii-rmi_preliminaries} restated]\label{lemma:ii-rmi}
A $k$-homogeneous distribution $\mu$ is {\iindependent} with parameter $D$ if and only if it is {\amindependent} with parameter $D$.
For an arbitrary $\tau =\{1\}^\Lambda$, $\Lambda \subset [n]$, consider the notation as described before, for the distribution $\mu_\tau$ conditional on always choosing the variables in $\Lambda$.
 $$ \norm{\pp_{u} - \pp}_1 \leq D, \quad \forall u$$ if and only if
$$ \expec{v}{|\pp_{v}-\pp|} \leq \frac{D}{k} \pp $$
\end{lemma}
\begin{proof}
For any $u,v \in [n]$, by Bayes rule we can compute the probability that an outcome sampled form $\mu$ contains both $u$ and $v$ as 
$$ \nu(u) p_u(v)= \frac{p(u)}{k} p_{u}(v) = \frac{p(v)}{k} p_{v}(u) = \nu(v) p_v(u)$$
Then for all $u$,
\begin{eqnarray*}
 \expec{v\sim \nu}{|p_{v}(u) - p(u)|} & =& \sum_{v} \nu(v) |p_{v}(u) - p(u)| \\
 & =& \sum_{v} |\nu(v) p_{v}(u) - \frac{p(v)}{k} p(u)| \\
 & =& \sum_{v} |\nu(u) p_{u}(v) - \frac{p(v)}{k} p(u)| \\
 & =& \sum_{v} \nu(u) |p_{u}(v) - p(v)| \\
 & =& \nu(u) \norm{\pp_{u} -\pp}_1 \\
\end{eqnarray*}
From this equality we deduce that $ \norm{\pp_{u} - \pp}_1 \leq D,  \forall u$ if and only if
$ \expec{v}{|\pp_{v}-\pp|} \leq D \qq $
\end{proof}

We conclude the section with the proof of the main theorem and the proof of Corollary~\ref{cor:homogenization}.

\begin{proof}[Proof of  Theorem~\ref{thm:main_theorem}.]
The trace of a matrix is the sum of its eigenvalues. The eigenvalues of the exponential of a matrix are the exponentials of the eigenvalues of the matrix. As a consequence, the trace of $ \tr \left( \exp( \theta \sum_i \xi_i \YY_i )\right)$ is an upper bound on the $\exp\left(\lambda_{\max} \left( \theta \sum_i \xi_i \YY_i )\right)\right)$.\\
Upper bound.
\begin{eqnarray*}
&&\prob{}{\lambda_{\max}\left( \sum_i \xi_i \YY_i - \expec{\xxi \sim \mu}{\sum_i \xi_i \YY_i} \right) \geq \delta \mu_{\max}} \\
& \leq & \prob{}{\tr \left( \exp\left( \theta \left( \sum_i \xi_i \YY_i - \expec{\xxi \sim \mu}{\sum_i \xi_i \YY_i} \right) \right)\right)  \geq \exp(\delta \theta \mu_{\max})} \\
& \leq & \frac{\expec{}{\tr \left( \exp\left( \theta \left( \sum_i \xi_i \YY_i -\expec{\xxi \sim \mu}{\sum_i \xi_i \YY_i} \right) \right) \right)}}{\exp(\delta \theta \mu_{\max})} \\
& \stackrel{\mbox{\small (i)}}{\leq} & \frac{\tr \left( \exp( 5D^2\theta^2 \expec{}{ \sum_i \xi_i \YY_i })\right)}{\exp(\delta \theta \mu_{\max})} \\
& \leq & \frac{\tr \left( \exp( 5D^2\theta^2 \mu_{\max}\II)\right)}{\exp(\delta \theta \mu_{\max})} \\
& \leq & d \exp(5 D^2 \theta^2 \mu_{\max} - \delta \theta\mu_{\max}) \\
& = & \exp\left(-\frac{\delta^2 \mu_{\max}}{20}\right) \quad \mbox{(choose $\theta=\delta/(10D)$)}
\end{eqnarray*}
In inequality (i) we applied Lemma~\ref{lemma:first_part} choosing $\HH = - \theta \expec{\xxi \sim \mu}{\sum_i \xi_i \YY_i}$ and Lemma~\ref{lemma:ii-rmi}.\\
Lower bound.
Since $\lambda_{\min}(-A)=-\lambda_{\max}(A)$
\begin{eqnarray*}
&& \prob{}{\lambda_{min}\left(\sum_i \xi_i \YY_i - \expec{\xxi \sim \mu}{\sum_i \xi_i \YY_i} \right) \leq  -\delta\mu_{\min}} \\
& = & \prob{}{\lambda_{max} \left( - \left( \sum_i \xi_i \YY_i -\expec{\xxi \sim \mu}{\sum_i \xi_i \YY_i }\right) \right)  \geq -(-\delta)\mu_{\min} } \\
& \leq  &  \frac{ \expec{}{\tr\left( \exp\left( - \theta  \left( \sum_i \xi_i \YY_i - \expec{\xxi \sim \mu}{\sum_i \xi_i \YY_i } \right) \right) \right)} }{\exp \left(\theta \delta \mu_{\min} \right) }  \\
& \leq & d \exp( 5 D^2 \theta^2 \mu_{\min}  -\theta \delta \mu_{\min}) \\
& = & d\exp\left(-\frac{\delta^2 \mu_{\min}}{20 D^2}\right)  \quad \mbox{(choose $\theta=\delta/(10D)$)}
\end{eqnarray*}
\end{proof}

\begin{proof}[Proof of Corollary~\ref{cor:homogenization}]
Let $\Omega = [2n] $ and $\muhom$ be the homogenization of $\mu$. We want to apply Theorem~\ref{thm:main_theorem} to the distribution $\muhom$. We notice that $\muhom$ is $n$-homogeneous, and we claim that $\muhom$ is {\iindependent} with parameter $2D$. 
In order to prove this, we notice that for all the outcomes $\sigma \subseteq \Omega$ such that $v\in \sigma$ and $(v+n) \in \sigma$ for some $v \in [n]$, $\muhom(\sigma)=0$. As a consequence,  for an arbitrary feasible conditioning $\tau =\{1\}^{\Lambda' \cup \Lambda ''}, \Lambda' \subset [n], \Lambda''\subset \{i+n ,i\in [n] \setminus \Lambda'\}$,  we have that:
\[
\prob{\muhom}{\xi_i=1|\xi_j =1 \forall j \in \Lambda' \land \xi_{j+n} = 1\forall j\in\Lambda'' } =  1 - \prob{\mu}{\xi_{i+n}=1|\xi_j =1 \forall j \in \Lambda' \land \xi_j = 0\forall j \in\Lambda'' },
\]
so, by definition of influence matrix:
$$ \calI_{\mu^\text{hom}}^{\Lambda}(i,j) = -\calI_{\mu^\text{hom}}^{\Lambda}(i,j+n).$$
Hence 
\begin{eqnarray}\label{eq:cor_a}
\sum_{v\in \Omega} |\calI_{\mu^\text{hom}}^{\Lambda}(i,v)| & = & \sum_{j\in [n]} |\calI_{\mu^\text{hom}}^{\Lambda}(i,j)| + \sum_{j\in [n]} |\calI_{\mu^\text{hom}}^{\Lambda}(i,j+n)| \\
& = & 2\sum_{j\in [n]} |\calI_{\mu^\text{hom}}^{\Lambda}(i,j)|. \nonumber
\end{eqnarray}
Furthermore, for an arbitrary feasible conditioning $\tau =\{1\}^{\Lambda' \cup \Lambda ''}, \Lambda' \subset [n], \Lambda''\subset \{i+n,i\in [n] \setminus \Lambda'\}$
\[
 \prob{\muhom}{\xi_i=1|\xi_j =1 \forall j \in \Lambda' \land \xi_{j+n} = 1\forall j\in\Lambda'' } =  \prob{\mu}{\xi_i=1|\xi_j =1 \forall j \in \Lambda' \land \xi_j = 0\forall j \in\Lambda'' },
\]
i.e the restriction to the first $n$ variables of $\mu^\text{hom}_{\tau}$ is distributed as $\mu_{\sigma_\Lambda}$ where $\Lambda = \Lambda' \cup \{i|j+n \in \Lambda''\}$ and $\sigma \in \{0,1\}^{\Lambda}$, $\sigma(i)=1\ \forall i\in \Lambda', \sigma(i)=0\ \forall i+n \in \Lambda''$.

 So, since $\mu$ is {\tsiindependent} with parameter $D$, for all $i\in [n]$:
\begin{eqnarray}\label{eq:cor_b}
\sum_{j\in [n]} |\calI_{\mu^\text{hom}}^{\Lambda}(i,j)| &\leq & \sum_{j\in [n]} \Psi_\mu^{\sigma_\Lambda}(i,j)\\
& \leq & D \nonumber
\end{eqnarray}
and
\begin{eqnarray}\label{eq:cor_c}
\sum_{j\in [n]} |\calI_{\mu^\text{hom}}^{\Lambda}(i+n,j)| &\leq & \sum_{j\in [n]} \Psi_\mu^{\sigma_\Lambda}(i,j)\\
& \leq & D. \nonumber
\end{eqnarray}
Combining equations (\ref{eq:cor_a}), (\ref{eq:cor_b}), and (\ref{eq:cor_c}) we obtain that $\mu^{\text{hom}}$ is {\iindependent} with parameter $2D$. Finally, we define the collection of matrices $\YY_{i} = \matzero$ for $i \in \setof{ n+1, \dots, 2n}$ and we notice that 
\[
\prob{\xxi \sim \mu}{\sum_{i\in n} \xi_i \YY_i} = \prob{\xxi \sim \mu^\text{hom}}{\sum_{i\in \Omega} \xi_i \YY_i}
\] and \[
\expec{\xxi \sim \mu}{\sum_{i\in n} \xi_i \YY_i} = \expec{\xxi \sim \mu^\text{hom}}{\sum_{i\in \Omega} \xi_i \YY_i}.
\]
 The corollary follows by applying Theorem~\ref{theorem:main_contribution} to the distribution $\mu^\text{hom}$ and the collection of matrices $\YY_i$ for $i\in \Omega$. 
\end{proof}

\printbibliography

\appendix
\section{Missing Proofs}\label{sec:proof_appendix}

\begin{proof}[Proof of Proposition~\ref{prop:SCP_iindep}.]
For every fixed $\tau \subset [n]$ and $v\in [n]\setminus \tau$ we have to prove that 
\[
\sum_{j\in [n]} |\expec{}{\xi_j| \xi_v=1 \land \xi_i=1\ \forall i \in \tau} - \expec{}{\xi_j|\xi_v=0 \land \xi_i=1\ \forall i \in \tau}| \leq 2
\]
We split the sum over $j$ into three parts.\\
For $j\in \tau$ we have that  $\sum_{j\in \tau} |\expec{}{\xi_j|\xi_v=1 \land \xi_i=1\ \forall i \in \tau} - \expec{}{\xi_j|\xi_v=0 \land \xi_i=1\ \forall i \in \tau}| =0 $.\\
For $j=v$, trivially $|\expec{}{\xi_v|\xi_v=1 \land \xi_i=1\ \forall i \in \tau} - \expec{}{\xi_v|\xi_v=0 \land \xi_i=1\ \forall i \in \tau}| = 1$.\\
For $j\in [n]\setminus (\tau \cup \{v\})$, let $\nu$ be the coupling between $\xi'$ and $\xi''$ given in the definition of SCP. The distribution $\nu$ over $\{0,1\}^n \times \{0,1\}^n$ satisfies that
$$\sum_{\xx} \nu(\xx,\yy) = \prob{}{\xi''=\yy},$$
$$\sum_{\yy} \nu(\xx,\yy) = \prob{}{\xi'=\xx}.$$
By the definitions of $\xi'$ and $\xi''$, 
\begin{eqnarray*}
\sum_{j\in [n]\setminus (\tau \cup \{v\})} |\expec{}{\xi_j|\xi_v =1 \land \xi_i=1\ \forall i \in \tau} - \expec{}{\xi_j|\xi_v=0 \land \xi_i=1\ \forall i \in \tau}| &=& \sum_{\xx,\yy} \norm{\xx-\yy}_1 \nu(\xx,\yy) \\
& \leq & \sum_{\xx,\yy} \nu(\xx,\yy) =1,
\end{eqnarray*}
where the inequality follows from the fact that $\nu(\xx,\yy)\neq 0$ only if $\xx$ can be obtained from the value of $\yy$ by either changing a single from 0 to 1 or by leaving all entries unchanged.\\
Combining the three parts we obtain 
$$ \sum_{j\in [n]} |\expec{}{\xi_j| \xi_v=1 \land \xi_i=1\ \forall i \in \tau} - \expec{}{\xi_j|\xi_v=0 \land \xi_i=1\ \forall i \in \tau}| \leq 2. $$
\end{proof}

\begin{proof}[Proof of Fact~\ref{fact:matrix_sqr}]
$$ 0\preceq (\AA-\BB)^2 = \AA^2-\AA\BB-\BB\AA+\BB^2 $$
Equivalently,
$$ \AA\BB+\BB\AA \preceq \AA^2+\BB^2 .$$
\end{proof}

\begin{proof}[Proof of Fact~\ref{fact:matrix_exp}]
Since $\AA\preceq \II$ and $\BB \succeq 0$, then $\AA-\BB\preceq \II$ hence
\begin{eqnarray*}
e^{\AA - \BB} & \preceq & \II + \left( \AA -  \BB \right) + (\AA-\BB)^2\\
& \preceq & \II + \AA -  \BB +2 \AA^2+2\BB^2\quad \mbox{(Fact~\ref{fact:matrix_sqr})}.
\end{eqnarray*}
\end{proof}

\section{Spectral Graph Sparsification using Random Spanning Trees}\label{sec:tree_appendix}

We define random spanning trees as in \cite{KS18}.  Let $G = (V,E,w)$ be a connected undirected weighted graph positive edge weights $w : E \to \R$.
  For each $e \in E$, we assign an arbitrary direction so that $e = (i,j)$, and we the vector $\bb_e \in \R^V$ to have all zero entries except $\bb_e(i) = 1$ and $\bb_e(j) = -1$.
  The Laplacian of $G$ is $\LL_G = \sum_{e \in E}  w(e) \bb_e \bb_e^\trp$. Let ${\cal T}_G$ denote the set of all spanning trees of $G$.
  We use $\MM^\dagger$ to denote the Moore-Penrose pseudo-inverse of a matrix $\MM$.
  We let $\PPi = \LL_G \LL_G^{\dagger}= \II - \vecone\vecone^{\trp}/n$ be the projection matrix to image of $\LL_G$.
  
\begin{definition}[$w$-uniform distribution on trees]
\label{def:wtUnifTreeDistr}
Let ${\cal D}_G$ be a probability distribution on ${\cal T}_G$ such that
\begin{align*}
\Pr_{X \sim {\cal D}_G }[ X = T ] \propto \prod_{e \in T} w(e).
\end{align*}
\end{definition}
We refer to ${\cal D}_G$ as the $w$-uniform distribution on ${\cal
  T}_G$. When the graph $G$ is unweighted, this corresponds to the
uniform distribution on ${\cal T}_G$.
It was shown in \cite{BBL09a} that the distribution of edges in random spanning trees are Strongly Rayleigh.
\begin{fact}[Spanning Trees are Strongly Rayleigh]
\label{fac:spantreemargin}
In a connected weighted graph $G$, the $w$-uniform distribution on
spanning trees is $(n-1)$-homogeneous Strongly Rayleigh.
\end{fact}

\begin{definition}[Effective Resistance]
The effective resistance of a pair of vertices $u,v \in V_G$ is defined as 
\begin{align*}
R_{\mathrm{eff}} (u,v) = \bb^\top_{u,v} L^\dagger \bb_{u,v},
\end{align*}
where $\bb_{u,v} \in \R^V$ is an all zero vector, except for entries of $1$ at $u$ and $-1$ at $v$.
\end{definition}

The following standard facts about random spanning
trees can be found in \cite{KS18}. 

\begin{definition}[Leverage Score]
The statistical leverage score, which we will abbreviate to leverage score, of an edge $e = (u,v) \in E_G$ is defined as 
\begin{align*}
l_e = w_e R_{\mathrm{eff}}(u,v).
\end{align*}
\end{definition}

\begin{fact}[Spanning Tree Marginals]
\label{fac:spantreemargin}
The probability $\Pr[e]$ that an edge $e\in E_G$ appears in a tree sampled $w$-uniformly randomly from ${\cal T}_G$ is given by
\begin{align*}
\Pr[e] = l_e,
\end{align*}
where $l_e$ is the leverage score of the edge $e$.
\end{fact}

\begin{proof}[Sketch of proof of Corollary~\ref{cor:spantreesparsify}]
  We briefly sketch this proof, which is identical to the proof of Theorem 1.3 in \cite{KS18}, except for a change in the number of spanning trees needed for concentration.

Let $T \subseteq E$ be a random spanning tree of $G$ in the sense of
Definition~\ref{def:wtUnifTreeDistr}.
Let the weights of the edges in $T$ be given by $w' : T \to R$ where 
$w'(e) = w(e) / l_e$, where $l_e$ is the leverage score of $e$ in $G$.
Thus the Laplacian of the tree is $\LL_T = \sum_{e \in T}  w'(e) b_e
b_e^\trp = \sum_{e \in T}  \frac{w(e)}{l_e} b_e b_e^\trp$.
Then by Fact~\ref{fac:spantreemargin}, $Pr[ e \in T ] = l_e$, and
hence $\E{\LL_T} = \LL_G$.

Note also that for all $e \in E$,  $\norm{ (\LL_G^\dagger)^{1/2} w(e)b_e
b_e^\trp (\LL_G^\dagger)^{1/2} }= l_e $.
Consider the random matrix $ (\LL_G^\dagger)^{1/2}  \LL_T
(\LL_G^\dagger)^{1/2}$.
The distribution of edge in the spanning tree can be seen as an $n-1$
homogeneous vector in $\{0,1\}^m$ where $m = |E|$.

Consider now a union of $t = C \epsilon^{-2} \log n$ spanning trees $T_{1}, \ldots, T_{t}$, and consider the vector $\xxi \in \setof{0,1}^{tm}$consisting of $t$ concatenated indicator vectors for the presence of edge $i$ in spanning tree $j$.
This random vector is $t(n-1)$-homogeneous and {\iindependent} with parameter $2$, because the concatenation of independent random binary vectors has {\iindependence} equal to the maximum among the concatenated parts.

To apply Theorem~\ref{theorem:main_contribution},
let $\xi_i$ be an entry of this vector corresponding to the indicator of $e \in E$ for some tree $j$.
Let
\[
\AA_i =
 (\LL_G^\dagger)^{1/2} w'(e)\bb_e
\bb_e^\trp (\LL_G^\dagger)^{1/2}
\]
Note $\AA_e \succeq 0$.
Now $|| \AA_e || = 1$ and
\[ \frac{1}{t} \E{\sum_{i} \xi_i \AA_i} = \E{(\LL_G^\dagger)^{1/2} \LL_T (\LL_G^\dagger)^{1/2} } =
    (\LL_G^\dagger)^{1/2} \LL_G (\LL_G^\dagger)^{1/2},\]
  Thus, as each we get $\lambda_{\max}(\frac{1}{t} \E{\sum_{i} \xi_i \AA_i}) = 1$.
  In the space\footnote{We omit a formal version of the argument restricting to matrices orthogonal to the kernel of all the matrices, but this is standard and straightforward.} orthogonal to $\vecone$, we have $\lambda_{\min}( \frac{1}{t} \E{\sum_{i} \xi_i \AA_i}) = 1.$
This means we can apply Theorem~\ref{theorem:main_contribution}, with $R = 1/t$,
$\mu = 1$ and conclude that whp.
\[
  \norm{(\LL_G^\dagger)^{1/2} \LL_T (\LL_G^\dagger)^{1/2} - \PPi } \leq \epsilon
  .
  \]
\end{proof}

\section{{\iindependence} does not imply {\tsiindependence}}
\label{sec:example_appendix}
In this section, we prove that {\tsiindependence} is strictly stronger than {\iindependence}, including for homogeneous distributions. We do this by showing an example of a $(k+1)$-homogeneous distribution over the set $[n+1]$. This distribution is one-sided {\iindependent} with parameter $D=O(1)$ but has unbounded {\tsiindependence} as $n$ and $k$ grow.

Let $n\geq k >0$ be two integers. Let $O$ a set that will be the set of all the possible outcomes. For every set $S\in \binom{[n]}{k}$, we add to $O$
\[
\begin{cases}
S\cup \setof{n+1} & \mbox {if } S \neq \setof{1,\dots,k} \\
S\cup \setof{k+1} & \mbox {if } S = \setof{1,\dots,k} \\ 
\end{cases}
\]
Finally we define $\mu$ as the uniform distribution over all the outcomes in $O$.

\begin{claim}\label{claim:B1}
For every $\Lambda \subset [n+1]$, 
$$ \norm{\calI_\mu^\Lambda}_{\infty} \leq 2\frac{n-k}{k} +1 $$
\end{claim}
\begin{proof}[Sketch of the proof.]
First we notice that also after conditioning on $\xi_j=1$ for some $j\in[n+1]$, the distribution will remain uniform. Furthermore, if we condition on $\xi_{n+1}=1$ we only remove the outcome $\setof{1,\dots,k,k+1}$ and if we condition on $\xi_j=1$ with $i\neq n+1$ and we restrict the distribution to $[n+1] \setminus \setof{j}$, we obtain again the distribution $\mu$ with $k$ and $n$ decreased by 1. 

As a consequence it suffices to compute $\norm{\calI_\mu^\Lambda}_{\infty}$ for $\Lambda = \emptyset$ and $\Lambda = \setof{n+1}$ and then we obtain $\norm{\calI_\mu^\Lambda}_{\infty}$ for all possible $\Lambda$ by induction. 
\end{proof}

\begin{claim}\label{claim:B2}
There exists $\Lambda \subset [n+1]$ and $\sigma \in \setof{0,1}^\Lambda$ such that, 
$$ \norm{\Psi_\mu^{\sigma_\Lambda}}_{\infty} \geq 2k \frac{n-k}{n} $$
\end{claim}
\begin{proof}[Sketch of the proof.]
Consider $\Lambda = \emptyset$. Notice that when we condition on $\xi_{n+1}=1$ we obtain the uniform distribution among all the outcomes in $O\setminus \setof{1,\dots,k,k+1}$ and when we condition on $\xi_{n+1}=0$ then $\setof{1,\dots,k,k+1}$ is the only possible outcome.
A straightforward calculation gives $$ \norm{\Psi_\mu^\emptyset} \geq 2k - 2\frac{k^2}{n} .$$
\end{proof}

In conclusion we notice that Claim~\ref{claim:B1} implies that $\mu$ is {\iindependent} with parameter 3, while if we chose for example $n= k^2$, Claim~\ref{claim:B2} tells us that it is not possible for $\mu$ to be {\tsiindependent} with parameter $D<k-1$.

\end{document}